\documentclass[11pt,english]{article}

\newcommand{\ARXIV}[1]{#1}
\newcommand{\ICALP}[1]{}

\usepackage{amsmath,amssymb,amsthm}
\usepackage{multicol}
\usepackage[table]{xcolor}
\usepackage[margin=1in]{geometry}
\usepackage{graphicx,color}
\usepackage{enumitem}
\usepackage{fullpage}
\usepackage[noblocks]{authblk}
\usepackage{tcolorbox}
\usepackage{babel}
\usepackage{wrapfig}
\usepackage{MnSymbol}
\usepackage{mdframed}
\usepackage{mathtools}
\usepackage{floatrow}
\usepackage{multirow}
\usepackage{booktabs}
\usepackage{palatino}

\usepackage[ruled,linesnumbered,vlined]{algorithm2e}
\usepackage{tablefootnote}
\usepackage{thm-restate}
\usepackage{bbding}
\usepackage{pifont}
\usepackage{nicefrac}
\DeclareMathOperator{\polylog}{polylog}
\DeclareMathOperator{\poly}{poly}

\newcounter{magicrownumbers}

\graphicspath{{./figures/}}

\usepackage{tablefootnote}

\definecolor{Darkblue}{rgb}{0,0,0.4}
\definecolor{Brown}{cmyk}{0,0.61,1.,0.60}
\definecolor{Purple}{cmyk}{0.45,0.86,0,0}
\definecolor{Darkgreen}{rgb}{0.133,0.543,0.133}

\usepackage[colorlinks,linkcolor=Darkblue,filecolor=blue,citecolor=blue,urlcolor=Darkblue,pagebackref]{hyperref}
\usepackage[nameinlink]{cleveref}

\usepackage[colorinlistoftodos,prependcaption,textsize=tiny]{todonotes}

\newcommand{\namedref}[2]{\hyperref[#2]{#1~\ref*{#2}}}
\newcommand{\propref}[1]{\hyperref[#1]{property~(\ref*{#1})}}

\newcommand{\eps}{\varepsilon}
\newcommand{\R}{\mathbb{R}}
\newcommand{\OO}{\widetilde{O}}
\newcommand{\CC}{\mathcal{C}}
\newcommand{\DS}{\mathcal{DS}}
\newcommand{\AAA}{\mathcal{A}}
\newcommand{\OPT}{\mbox{\rm OPT}}
\newcommand{\IGNORE}[1]{}

\newtheorem*{theorem*}{Theorem}
\newtheorem{theorem}{Theorem}
\newtheorem{lemma}{Lemma}

\newtheorem{corollary}{Corollary}

\newtheorem{problem}{Problem}
\newtheorem*{question*}{Question}

\newtheorem*{conjecture*}{Conjecture}

\newcommand{\old}[1]{{}}

\hypersetup{
    colorlinks=true,
    linkcolor=blue,
    filecolor=violet,  
    citecolor=violet,
    urlcolor=cyan,
    pdfpagemode=FullScreen,
}

\title{Fully Dynamic Geometric Vertex Cover and Matching}

\date{}

 \author{Sujoy Bhore\thanks{Department of Computer Science \& Engineering, Indian Institute of Technology Bombay, Mumbai, India.\\ Email: sujoy@cse.iitb.ac.in}
 \quad
 Timothy M. Chan\thanks{Department of Computer Science, University of Illinois at Urbana-Champaign. Email: tmc@illinois.edu.  Work supported in part by
 NSF Grant CCF-2224271.}
 }

\begin{document}
\maketitle

\begin{abstract}
In this work, we study two fundamental graph optimization problems, minimum vertex cover (MVC) and maximum-cardinality matching (MCM), for intersection graphs of geometric objects, e.g., disks, rectangles, hypercubes, etc., in $d$-dimensional Euclidean space. We consider the problems in fully dynamic settings, allowing insertions and deletions of objects. 

We develop a general framework for dynamic MVC in intersection graphs, achieving sublinear amortized update time for most natural families of geometric objects. 
In particular, we show that -

\begin{itemize}
    \item For a dynamic collection of disks in $\mathbb{R}^2$ or hypercubes in $\mathbb{R}^d$ (for constant $d$), it is possible to maintain a $(1+\varepsilon)$-approximate vertex cover in $\polylog$ amortized update time.
    These results also hold in the bipartite case. 

    \item For a dynamic collection of rectangles in $\mathbb{R}^2$, it is possible to maintain a 
    $(\frac{3}{2}+\varepsilon)$-approximate vertex cover in $\polylog$ amortized update time.
\end{itemize}
Along the way, we obtain the first near-linear time static algorithms for MVC in the above two cases with the same approximation factors.

Next, we turn our attention to the MCM problem.
Although our MVC algorithms automatically
allow us to approximate the size of the MCM in bipartite geometric intersection graphs, they do not produce a
matching. We give another general framework to maintain an approximate maximum matching, and 
further extend the approach to handle non-bipartite intersection graphs.
In particular, we show that -

\begin{itemize}
     \item For a dynamic collection of (bichromatic or monochromatic) disks in $\mathbb{R}^2$ or hypercubes in $\mathbb{R}^d$ (for constant $d$), it is possible to maintain a $(1+\varepsilon)$-approximate matching in $\polylog$ amortized update time. 
     
\end{itemize}

\end{abstract}


\section{Introduction}


Dynamic geometric algorithms form a vital area of research in Computational Geometry.
In dynamic settings, the input is subject to discrete changes over time, that is, insertions and deletions of geometric objects.
Over the years,
dynamic algorithms have been developed 
for a variety of geometric problems, including
convex hull~\cite{OvermarsL81, Chan01, BrodalJ02, Chan10, Chan20a}, 
closest/farthest pair~\cite{Eppstein95,Chan20},
connectivity~\cite{Chan06, AfshaniC09, ChanPR11, KaplanKKKMRS22}, 
width and other measures~\cite{Chan01,Chan03b,Chan20a},
coresets~\cite{Chan09},
and many others.  
More recently, much attention has turned towards {\sf NP}-hard geometric optimization problems---building dynamic data structures that can maintain solutions for an optimization task at hand, in truly sublinear update time. 
Since it is only possible to maintain some approximate solution for an {\sf NP}-hard problem, the broad goal is to explore the trade-offs between update time and approximation ratio. Optimization problems that are solvable in polynomial time but do not have  subquadratic exact algorithms may also benefit from allowing approximation.  Several fundamental geometric optimization problems have been studied in dynamic settings over the last few years, for instance, set cover~\cite{AgarwalCSXX22,ChanH21,ChanHSX22,KLRSW23}, piercing or hitting set~\cite{AgarwalCSXX22,AHRS23}, and
independent set~\cite{bhore2020dynamic, Henzinger0W20, BLN22, BNTW23}.

In this work, we focus on two other fundamental optimization problems: minimum vertex cover (MVC) and maximum-cardinality matching (MCM\@). Let $G=(V, E)$ be an undirected graph. A subset of vertices $S\subseteq V$ is a vertex cover if each edge has at least one endpoint in $S$. An MVC refers to a vertex cover with the minimum cardinality.
A simple algorithm that greedily computes a maximal matching\footnote{A matching in an input graph $G$ is a subset of edges $M\subseteq E$ such that no two edges in $M$ share a common endpoint. A matching $M$ is maximal if for every edge $(u,v)\in E\setminus M$,  either $u$ or $v$ is matched in $M$.} and outputs its vertices gives a $2$-approximation for MVC.  This is likely to be the best possible for general graphs, as the problem is hard to approximate within factor $2-o(1)$ 
under the unique games conjecture~\cite{KhotR08}.

MVC stands out as one of the most important problems, due to its connection to other fundamental problems.
For example, in the exact setting, the
maximum independent set (MIS) problem is equivalent to MVC, as the complement of an MIS
is an MVC\@. However, from the approximation perspective, a sharp dichotomy exists between these problems; for instance,  MIS on general graphs is much harder to approximate\footnote{No polynomial time algorithm can achieve an approximation factor $n^{1-\varepsilon}$, 
unless $\textsf{P}=\textsf{ZPP}$~\cite{Zuckerman07}.  Throughout this paper, $\eps$ denotes an arbitrarily small positive constant.}. Another closely related classical optimization on graphs is MCM, where the objective is to find a matching with the largest number of edges. The size of an MVC is obviously at least as large as the size of an MCM; for bipartite graphs, they are well-known to be equal.
MCM is polynomial-time solvable:
the classical algorithm by Hopcroft and Karp~\cite{HopcroftK73} runs in $O(m\sqrt{n})$ time
for bipartite graphs  with $n$ vertices and $m$ edges,
and Vazirani's algorithm~\cite{Vazirani94} achieves
the same run time for general graphs. By recent breakthrough results~\cite{ChenKLPGS22}, MCM can be solved in $m^{1+o(1)}$ time for bipartite graphs. 
Earlier, Duan and Pettie~\cite{DuanP10} obtained $O(m)$-time $(1+\eps)$-approximation algorithms for general graphs.

\paragraph{Geometric MVC \& MCM.} We focus on the intersection graph of a collection $\mathcal{L}=\{\ell_1,\ldots,\ell_n\}$ of geometric objects in $d$-dimensional Euclidean space $\R^d$ (for $d$ constant). In the intersection graph $G_{\mathcal{L}} = (V_{\mathcal{L}},E_{\mathcal{L}})$ of $\mathcal L$, each object $\ell_i\in\mathcal{L}$ is represented by a vertex $v_i\in V_{\mathcal{L}}$, and any pair of intersecting objects $\ell_i,\ell_j\in \mathcal{L}$ 
corresponds to an edge $(v_i,v_j)\in E_{\mathcal{L}}$. A wide range of optimization problems have been studied over the years for various families of geometric intersection graphs,
and better results are often possible in such geometric graphs.
Concerning MVC, the problem remains \textsf{NP}-hard, even for intersection graphs of unit disks~\cite{clark1990unit}, but there exist \textsf{PTAS}s for intersection graphs of disks, squares, or other ``fat'' objects in any constant dimension~\cite{erlebach2005polynomial}.
For rectangles\footnote{All rectangles and boxes are axis-aligned by default in this paper.} in the plane,
Bar-Yehuda, Hermelin, and Rawitz~\cite{bar2011minimum}
obtained a $(\frac{3}{2}+\varepsilon)$-approximation algorithm.
Bar-Yehuda et al.'s work
made use of Nemhauser and Trotter's standard LP-based kernelization for vertex cover~\cite{nemhauser1975vertex}, which allows us to approximate the MVC by flipping to an MIS instance.
Following the same kernelization approach,
Har-Peled~\cite{SarielVC23} noted that by using
the known \textsf{QQPTAS}\footnote{We refer to a $(1 + \varepsilon)$-approximation algorithm with running time of the form $n^{O(\polylog n)}$ or 
$n^{O(\poly(\log \log n))}$, for any fixed $\eps>0$, as a \textsf{QPTAS} or \textsf{QQPTAS} respectively.} for MIS for rectangles~\cite{ChuzhoyE16}, one can immediately obtain a \textsf{QQPTAS} for MVC for rectangles.
Similarly, by using the known \textsf{QPTAS}
for MIS for  polygons in the plane~\cite{AdamaszekHW19}, one can obtain 
a {\sf QPTAS} for MVC for  polygons.


MCM on geometric intersection graphs has also been studied.
Efrat, Itai, and Katz~\cite{efrat2001geometry} showed how to compute the maximum matching in bipartite unit disk graphs in $O(n^{3/2} \log n)$ time. Their algorithm works for other geometric objects; for example, it runs in $O(n^{3/2} \polylog n)$ time
for bipartite intersection graphs of arbitrary disks,
by using known dynamic data structures for disk intersection searching~\cite{KaplanKKKMRS22}.
Bonnet, Cabello, and Mulzer~\cite{bonnet2023maximum} improved
the running time to $O(n^{\omega/2})$ for certain cases, e.g., translates of a convex object, or
geometric objects with low ``density'',
where $\omega<2.38$ is the matrix multiplication exponent.
Recently, Har-Peled and Yang~\cite{Har-PeledY22} presented near-linear time 
$(1+\varepsilon)$-approximation algorithms for MCM in (bipartite or non-bipartite) intersection graphs of arbitrary disks, among other things. 

\paragraph{Dynamic geometric MVC \& MCM.}  In this paper,
we are interested in the dynamic version of MVC and MCM problems, for geometric intersection graphs, subject to insertions and deletions of objects in $\mathcal{L}$. The goal is to obtain a data structure that can efficiently maintain an approximate MVC or MCM in $G_{\mathcal{L}}$ at any point in time, with truly sublinear update times for both insertions and deletions. 

\begin{tcolorbox}
{\begin{problem}\label{prob1}
    Given a dynamic collection of geometric objects in  $\R^d$, is it possible to maintain a $(1+\varepsilon)$-approximate MVC in truly sublinear update time? 
\end{problem}}
\end{tcolorbox}


\begin{tcolorbox}
{\begin{problem}\label{prob2}
    Given a dynamic collection of geometric objects in  $\R^d$, is it possible to maintain a $(1+\varepsilon)$-approximate MCM in truly sublinear update time? 
\end{problem}}
\end{tcolorbox}

Dynamic MVC \& MCM are well-studied problems in the dynamic graph algorithms literature, under edge updates; see \cite{OnakR10, GuptaP13, BernsteinS15, BernsteinS16, PelegS16, Solomon16, BhattacharyaHI18, BhattacharyaK19, Behnezhad23, BhattacharyaKSW23, AzarmehrBR24} for a non-exhaustive list.
However, none of these results particularly apply to our scenario, since the insertion/deletion of a single object corresponds to a vertex update, which may require many edge updates.
A natural barrier for designing dynamic graph algorithms under vertex updates is the maximum degree of the input graph, which may be $\Omega(n)$. For geometric intersection graphs, however, the challenge is to break such barriers by not explicitly maintaining the graph itself, which may have $m=\Omega(n^2)$ size in the worst case.

Although geometric versions of dynamic set cover, dynamic hitting set, and dynamic independent set have received much attention recently (as we have mentioned), we are not aware of any prior work on dynamic geometric vertex cover or dynamic geometric matching.

For  unit disks in the plane, or more generally, fat objects of roughly the same size in $\mathbb{R}^d$ (for constant $d$), it is not difficult to maintain a $(1+\varepsilon)$-approximate MVC in constant update time, by using the standard shifted-grid technique~\cite{hochbaum1985approximation}. Moreover, this technique extends to  maintain a $(1+\varepsilon)$-approximate non-bipartite  MCM in 
constant update time (see also \cite{Har-PeledY22} in the static setting). 
However, the problems become much more challenging for more general families of objects, such as disks of arbitrary radii, or arbitrary (non-fat) rectangles.





In this work, we resolve Problems~\ref{prob1} and~\ref{prob2} in an affirmative sense for 
many types of geometric objects:

\paragraph{Our contributions to MVC.}


We develop a general framework for dynamic geometric MVC (Theorem~\ref{thm:main}), which implies
all the specific results listed in Table~\ref{table1}.





\begin{table}[ht]
\begin{tabular}{llll}
\toprule

\multicolumn{1}{c}{Objects}  & \multicolumn{1}{c}{Approx. Ratio} & \multicolumn{1}{c}{Amortized update time} & \multicolumn{1}{c}{Reference} \\ 

\midrule

\rowcolor[gray]{.9} Disks in $\mathbb{R}^2$ & $1+\varepsilon$ & $O(2^{O(1/\eps^2)}\log^{O(1)}n)$  & Corollary~\ref{disk:MVC} \\ 

Rectangles in $\mathbb{R}^2$ & $\frac{3}{2}+\eps$ & $O(2^{O(1/\eps^2)}\log^{O(1)}n)$ & Corollary~\ref{cor:rect} \\ 

\rowcolor[gray]{.9}  Fat boxes in $\mathbb{R}^d$  & $1+\varepsilon$ & $O(2^{O(1/\eps^d)}\log^{O(1)}n)$ & Corollary~\ref{cor:fat:MVC} \\ 

Bipartite disks in $\mathbb{R}^2$ & $1+\varepsilon$ & $O((1/\eps^7)\log^{O(1)}n)$ & Corollary~\ref{MVC:bipartite:disk} \\ 

\rowcolor[gray]{.9} Bipartite boxes (non-fat) in $\mathbb{R}^d$ & $1+\eps$ & $O((1/\eps^7)\log^{O(1)}n)$ & Corollary~\ref{MVC:bipartite:box} \\ 

\bottomrule
\end{tabular}
\caption{\small Summary of results on dynamic MVC for intersection graphs of geometric objects.\label{table1}}
\end{table}

These results on geometric vertex cover are notable when compared against recent results on related dynamic geometric problems such as set cover and independent set:
\begin{itemize}
\item The \emph{polylogarithmic} update times in Table~\ref{table1} are significantly better than known results for dynamic geometric set cover.  For example, even for the simple case of squares in $\R^2$, the best update time for set cover was near $\sqrt{n}$~\cite{ChanHSX22}; for the case of disks in $\R^2$, the best update time was near $n^{12/13}$~\cite{ChanH21} (and was only for maintaining the approximate size but not a cover).
\item Many of the results in Table~\ref{table1} achieve the \emph{ideal
approximation ratio}, namely, $1+\varepsilon$.  In contrast, previous
dynamic algorithms for geometric set cover~\cite{ChanH21,ChanHSX22} were all in the regime of $O(1)$ approximation ratio.  Similarly, for dynamic geometric MIS, recent algorithms for squares or hypercubes~\cite{Henzinger0W20} and disks~\cite{BNTW23} all achieved $O(1)$ approximation only.
\item Our framework is \emph{general} and applicable to more types of objects.
Though not stated in the table, we can handle any class of semialgebraic objects with constant description complexity in sublinear time with approximation ratio $2+\varepsilon$ (or $1+\eps$ in the bipartite cases).  In contrast, the known results for dynamic geometric set cover~\cite{ChanH21,ChanHSX22} were limited to certain specific types of objects; for example, they did not extend to balls in $\R^3$ or rectangles in $\R^2$.  There were no nontrivial results for dynamic MIS for rectangles with constant approximation ratio.
\end{itemize}

A prerequisite to a $(1+\varepsilon)$-approximate dynamic algorithm with polylogarithmic update time is the existence of a near-linear-time \emph{static} algorithm.
MIS for disks (or fat objects)  admits static {\sf PTAS}s~\cite{erlebach2005polynomial,Chan03}, but not an {\sf EPTAS}\footnote{An {\sf EPTAS} refers to a {\sf PTAS} with running time $O(n^c)$ for a constant $c$ independent of $\eps$.} under standard hypotheses in parameterized complexity~\cite{Marx05}; 
similarly, set cover for disks admits static {\sf PTAS}s~\cite{MustafaR10} but not an {\sf EPTAS}~\cite{Marx05}.
In contrast, an {\sf EPTAS} for MVC for disks (or fat objects) can be obtained by known techniques using the aforementioned Nemhauser--Trotter kernel~\cite{nemhauser1975vertex}.  This explains why in some sense, geometric vertex cover is actually a better problem to study in dynamic settings than geometric independent set or  geometric set cover.

However, a static {\sf EPTAS}
for MVC for disks (or fat objects) that runs in \emph{near linear time} has not been explicitly given before, to the best of our knowledge.
Similarly, the static $(\frac32 +\eps)$-approximation algorithm
for MVC for rectangles by Bar-Yehuda, Heremelin, and Rawitz~\cite{bar2011minimum} runs in polynomial time but not in near linear time.  In deriving our dynamic results, we obtain
the first static near-linear-time algorithms in these cases, which may be of independent interest.

\old{
First, we show that, for a dynamic collection of disks of arbitrary radii in $\mathbb{R}^2$, we show that it is possible to maintain $(1+\eps)$-approximate MVC in $\polylog$ amortized update time (Corollary~\ref{disk:MVC}). 

Next, we show that, for a dynamic collection of rectangles in $\mathbb{R}^2$, it is possible to maintain $(\frac{3}{2}+\varepsilon)$-approximate MVC in $\polylog$ amortized update time (Corollary~\ref{cor:rect}). Moreover, we show that, for fat axis-aligned boxes in $\mathbb{R}^d$ for any constant $d$, it is possible to maintain $(1+\varepsilon)$-approximate MVC in $\polylog$ amortized update time (Corollary~\ref{cor:fat:MVC}). 

Finally, we extend the results to the bipartite version as well. 
We show that, for a bipartite set of disks in $\mathbb{R}^2$, we can maintain $(1+\varepsilon)$-approximate MVC in $\polylog$-amortized update time (Corollary~\ref{MVC:bipartite:disk}). 
For the case of a bipartite intersection graph between two sets of (not necessarily fat) boxes in $\mathbb{R}^d$, we show that it is possible to maintain $(1+\varepsilon)$-approximate MVC in $\polylog$-amortized update time (Corollary~\ref{MVC:bipartite:box}). 
}


\paragraph{Our contributions to MCM.}
Next, we focus on MCM\@. 
Our approximation algorithms for MVC in
bipartite geometric intersection graphs allows us to automatically approximate the size of the MCM in such bipartite graphs. However, it does not compute a matching. 
We develop another general framework for dynamic geometric MCM  (Theorem~\ref{thm:main:mcm}), and further extend it to non-bipartite intersection graphs.
The framework implies the specific results in Table~\ref{table2}:



\begin{table}[ht]
\begin{tabular}{llll}
\toprule

\multicolumn{1}{c}{Objects}  & \multicolumn{1}{c}{Approx. Ratio} & \multicolumn{1}{c}{Amortized update time} & \multicolumn{1}{c}{Reference} \\ 

\midrule

\rowcolor[gray]{.9} Bipartite  disks in $\mathbb{R}^2$ & $1+\varepsilon$ & $O((1/\eps^3)\log^{O(1)}n)$  & Corollary~\ref{cor:mcm:disk} \\ 

Bipartite boxes in $\mathbb{R}^d$ & $1+\eps$ & $O((1/\eps^3)\log^{O(1)}n)$ & Corollary~\ref{cor:mcm:box} \\ 

\rowcolor[gray]{.9} Disks in $\mathbb{R}^2$  & $1+\varepsilon$ & $O(2^{O(1/\eps)}\log^{O(1)}n)$  & Corollary~\ref{cor:mcm:disk2} \\ 

Boxes in $\mathbb{R}^d$ & $1+\varepsilon$ & $O(2^{O(1/\eps)}\log^{O(1)}n)$  & Corollary~\ref{cor:mcm:box2} \\ 

\bottomrule
\end{tabular}
\caption{\small Summary of results on dynamic MCM for intersection graphs of geometric objects.\label{table2}}
\end{table}

In particular, we succeed in efficiently dynamizing the known static near-linear-time $(1+\eps)$-approximation MCM algorithms for (bipartite or non-bipartite) disks by Har-Peled and Yang~\cite{Har-PeledY22}.

\old{
we show that for a bipartite set of disks in $\mathbb{R}^2$, 
it is possible to maintain $(1+\varepsilon)$-approximate MCM in $\polylog$-amortized update time (Corollary~\ref{cor:mcm:disk}). 
Moreover, for a bipartite set of axis-aligned boxes in $\mathbb{R}^d$, we can maintain $(1+\varepsilon)$-approximate MCM in $\polylog$-amortized update time (Corollary~\ref{cor:mcm:box}).
Finally, we adapt our results from the bipartite MCM
case to the general MCM. Specifically, we show that similar guarantees can be obtained for disks in $\mathbb{R}^2$ (Corollary~\ref{cor:mcm:disk2}) and axis-aligned boxes in $\mathbb{R}^d$ (Corollary~\ref{cor:mcm:box2}). 
}




\paragraph{Our framework and techniques for MVC.}
Our framework for MVC has only 2 requirements on the geometric object family: 
\begin{enumerate}
    \item[(i)]~the existence of an efficient data structure for \emph{dynamic intersection detection} (maintaining a set of input objects under insertions and deletions so that we can quickly decide whether a query object intersects any of the input objects), and 
    \item[(ii)]~the existence of an efficient \emph{static approximation algorithm}.
\end{enumerate}
If the data structure in (i) takes polylogarithmic query and update time and the algorithm in (ii) takes near linear time ignoring polylogarithmic factors, then the resulting dynamic algorithm achieves polylogarithmic (amortized) update time, and its approximation ratio
is the same as the static algorithm, up to $1+\eps$ factors.

In some sense, the framework is as general as possible:  Requirement (ii) is an obvious prerequisite (as we have already noted); in fact, we only require a static algorithm for the case when the optimal value is promised to be at least $n/2$ roughly.
Requirement (i) is also natural, as some kind of geometric range searching data structures is provably necessary.   For example, consider the case of disks in $\R^2$.  Any dynamic approximate MVC algorithm needs to recognize whether the MVC size is zero, and so must
know whether the intersection graph is empty.   Dynamic disk range emptiness (maintaining a set of input points under insertions and deletions so that we can quickly decide whether a query disk contains any of the input points) can be reduced to this problem, by inserting all the input points, and repeatedly inserting a query disk and deleting it.  (This explains why for our result on disks, we need multiple logarithmic factors in the update time, since the best algorithm for dynamic disk range emptiness requires multiple logarithmic factors~\cite{Chan10}, although the current best algorithm for dynamic disk intersection detection requires still more logarithmic factors~\cite{KaplanKKKMRS22}.)

Our framework for MVC is established by combining a number of ideas.  The initial idea is a standard one: \emph{periodic rebuilding}.  The intuition is that when the optimal value is large, we can afford to do nothing for a long period of time, without hurting the approximation ratio by much.  This idea has been used before in a number of recent works on dynamic geometric independent set (e.g., \cite{CardinalIK21}),
dynamic piercing (e.g., \cite{AHRS23}), and 
dynamic graph algorithms (e.g., \cite{GuptaP13}), although the general approach of periodic rebuilding is commonplace and appeared in much earlier works in dynamic computational geometry (e.g., \cite{Chan01,Chan03b}).  

The main challenge now lies in the case when the optimal value $\OPT$ is small. 
Our idea is to compute Nemhauser and Trotter's standard LP-based kernel for MVC, and then run a static algorithm on the kernel.
Unfortunately, we are unable to solve the LP exactly in sublinear time.  However, by using the \emph{multiplicative weight update} (MWU) method,
we show that it is possible to solve the LP approximately in time roughly linear in $\OPT$ (and thus sublinear in $n$ in the small $\OPT$ case), by using the intersection detection data structure provided by (i).  We further show that an approximate solution to the LP is still sufficient to yield a kernel with approximately the same size.
The idea of using MWU in the design of dynamic geometric algorithms was pioneered in Chan and He's 
recent work on dynamic geometric set cover \cite{ChanH21}, and also appeared in subsequent work on dynamic piercing \cite{AHRS23}; 
however, these works solved the LP in time superlinear in $\OPT$, which is why their final update time bounds
were significantly worse.

In geometric approximation algorithms, the best-known application of MWU is geometric set cover or hitting set \cite{BronnimannG95}.
Our work highlights its usefulness also to geometric vertex cover---ironically, the application to vertex cover is even simpler (and so
our work might have additional pedagogical value).  Furthermore, the efficient implementation of MWU for dynamic geometric vertex cover
turns out to lead to a nice, unexpected application of another known technique: namely, Eppstein's data structure technique for \emph{generalized dynamic closest pair} problems~\cite{Eppstein95} (see also \cite{Chan20}).

\paragraph{Our framework and techniques for MCM.}
For MCM, our framework only needs requirement~(i).  Here, we need different ideas, since kernels for matching do not seem to work as efficiently.  In the bipartite case, we use an approach based on Hopcroft and Karp's classical matching algorithm~\cite{HopcroftK73}, which is known to
yield good approximation after a constant number of iterations.  We show how to implement the approximate version of Hopcroft and Karp's algorithm in time roughly linear in $\OPT$, by using the data structure from (i).  Previously, Efrat, Itai, and Katz~\cite{efrat2001geometry} showed how to
implement Hopcroft and Karp's algorithm faster using geometric data structures, but their focus was on static algorithms whereas our goal is in getting sublinear time.  Thus, our adaptation of Hopcroft--Karp will be a little more delicate.

In the general non-bipartite case, we apply a technique by Lotker, Patt-Shamir, and Pettie~\cite{LotkerPP15}. which reduces non-bipartite to the bipartite case in the approximate setting.  Har-Peled and Yang~\cite{Har-PeledY22} also applied the same technique to derive their static approximation algorithms for geometric non-bipartite MCM\@.  We reinterpret this technique in terms of \emph{color-coding}~\cite{AlonYZ95}, which allows for efficient derandomization and dynamization, as well as a simpler analysis.



\section{Minimum Vertex Cover}
In this section, we study the MVC problem for dynamic intersection graphs of geometric objects.

\subsection{Approximating the LP via MWU}

\newcommand{\updateweight}{\textsc{update-weight}}
\newcommand{\findminedge}{\textsc{find-min-weight-edge}}

For a graph $G=(V,E)$, a \emph{fractional vertex cover} is a vector
$(x_v)_{v\in V}$ such that $x_u+x_v\le 1$ for all $uv\in E$ and $x_v\in [0,1]$ for all $v\in V$.  
Its \emph{size} is defined as $\sum_{v\in V}x_v$.  Finding a minimum-size fractional vertex cover
corresponds to solving an LP, namely, the standard LP relaxation of the MVC problem.

It is known that this LP is equivalent to computing the MVC in a related bipartite graph,
and thus can be solved exactly by known bipartite MCM algorithms---in fact,
in time almost linear in the number of edges by recent breakthrough results~\cite{ChenKLPGS22}. 
However, there are two issues that prevent us from applying such algorithms.
First of all, we are considering geometric intersection graphs, which may have $\Omega(n^2)$ number of edges; this issue could potentially be fixed by
using known techniques involving \emph{biclique covers} to sparsify the graph (maximum matching in a bipartite graph
then reduces to maximum flow in a sparser 3-layer graph~\cite{FederM95}). Second, we want efficient data
structures that can solve the LP still faster, in \emph{sublinear} time when $\OPT$ is small.

For our purposes, we only need to solve the LP approximately.
Our idea is to use a different well-known technique: \emph{multiplicative weight update} (MWU)\@. The key lemma is stated below.  The MWU algorithm and analysis here are nothing new (the description is short enough that we opt to include it to be self-contained),
and such algorithms have been used before for static and dynamic geometric set cover and other geometric optimization
problems (the application to vertex cover turns out to be a little simpler). However, our contribution is not in the proof of the lemma, but in the realization that MWU reduces the problem to designing a dynamic data structure (for finding min-weight edges subject to vertex-weight updates), which geometric intersection graphs happen to possess, as we will see.

\begin{lemma}\label{lem:mwu}
We are given a graph $G=(V,E)$.  Suppose there is a data structure $\DS$ for storing a vector $(w_v)_{v\in V}$
that can support the following two operations in $\tau$ time: (i) find an edge $uv\in E$ minimizing $w_u+w_v$, and (ii) update a number $w_v$.

Given a data structure $\DS$ for the vector that currently has $w_v=1$ for all $v\in V$,
we can compute a $(1+O(\delta))$-approximation to the minimum fractional vertex cover in
$\OO((1/\delta^2)\,\OPT\cdot \tau)$ time.  Here, $\OPT$ denotes the minimum vertex cover size.
\end{lemma}
\begin{proof}
Given a number $z$, the following algorithm attempts to find a fractional vertex cover of size at most~$z$ (below, $W$ denotes $\sum_{v\in V}w_v$):

\begin{quote}
\begin{tabbing}
for \= \kill
let $w_v=1$ for all $v\in V$, and $W=n$\\
while there exists $uv\in E$ with $w_u+w_v < W/z$ do \\
\> let $uv$ be such an edge\\
\> $W\leftarrow W + \delta(w_u+w_v)$,\ $w_u\leftarrow (1+\delta)w_u$,\ $w_v\leftarrow (1+\delta)w_v$
\end{tabbing}
\end{quote}

If and when the algorithm terminates, we have $w_u+w_v\ge W/z$ for all $uv\in E$.
Thus, defining $x_v := \min\{zw_v/W,\,1\}$, we have $x_u+x_v\ge 1$ for all $uv\in E$, and $\sum_{v\in V} x_v \le z$, i.e.,
$(x_v)_{v\in V}$ is a fractional vertex cover of size at most $z$.

We now bound the number of iterations $t$.
In each iteration, $W$ increases by at most a factor of $1+\delta/z$.  Thus,
at the end, 
\[ W\le (1+\delta/z)^t n.\]
Write each $w_v$ as $(1+\delta)^{c_v}$ for some integer~$c_v$.
Let $(x_v^*)_{v\in V}$ be an optimal fractional vertex cover of size $z^*$.
In each iteration, $\sum_{v\in V} c_vx_v^*$ increases by at least 1 (since we increment $c_u$ and $c_v$
for the chosen edge $uv$, and we know $x_u^*+x_v^*\ge 1$).
Thus, at the end, $\sum_{v\in V} c_vx_v^*\ge t$.
Since $\sum_{v\in V} x_v^* = z^*$,
it follows that $\max_{v\in V} c_v\ge t/z^*$.  Thus,
\[ W\ge (1+\delta)^{t/z^*}.\]
Therefore, $(1+\delta)^{t/z^*}\le (1+\delta/z)^tn\le e^{\delta t/z}n$, implying $(t/z^*)\ln (1+\delta) \le \delta t/z + \ln n$.
So, if $z\ge (1+\delta)z^*$, then $t\le \frac{z^*\ln n}{\ln (1+\delta) - \delta/(1+\delta)}= O((1/\delta^2)z^*\log n)$.

Note that only $O(t)=O((1/\delta^2)z^*\log n)$ of the numbers $w_v$ are not equal to 1, so 
the vector $(x_v)_{v\in V}$ can be encoded in $\OO(z^*)$ space.

We can try different $z$ values by binary or exponential search
till the algorithm terminates in $O((1/\delta^2)z^*\log n)$
iterations.  Each run can be implemented with $O((1/\delta^2)z^*\log n)$ operations in $\DS$.
After each run, we reset all the modified values $w_v$ back to 1 by $O((1/\delta^2)z^*\log n)$ update operations in $\DS$.
\end{proof}

\subsection{Kernel via Approximate LP}

Our approach for solving the vertex cover problem is to use a standard \emph{kernel} by
Nemhauser and Trotter~\cite{nemhauser1975vertex}, which allows us to reduce the problem to an instance where the number of vertices
is at most $2\,\OPT$.  Nemhauser and Trotter's construction is obtained from the LP solution: the kernel is 
simply the subset of all vertices $v\in V$ with $x_v=\frac12$.

In our scenario, we are only able to solve the LP approximately.  We observe that this is still enough to give a kernel
of approximately the same size.  We adapt the standard analysis of Nemhauser and Trotter, but
some extra ideas are needed.  The proof below will be self-contained.

\begin{lemma}\label{lem:kernel}
Let $c\ge 1$ and $0\le\delta<\gamma<\frac14$.
Given a $(1+O(\delta))$-approximation to the minimum fractional vertex cover in a graph $G=(V,E)$,
we can compute a subset $K\subseteq V$ of size at most $(2+O(\gamma))\,\OPT$, in $O(\OPT)$ time, such that 
a $c(1+O(\sqrt{\delta/\gamma}))$-approximation to the minimum vertex cover of $G$ can be obtained from
a $c$-approximation to the minimum vertex cover of $G[K]$ (the subgraph of $G$ induced by $K$).  Here, $\OPT$ denotes the minimum vertex cover size of $G$.
\end{lemma}
\begin{proof}
Let $(x_v)_{v\in V}$ be the given fractional vertex cover.  Let $\lambda < \gamma$ be a parameter to be set later.
Pick a value $\alpha\in [\frac12 - \gamma-\lambda,\frac12-\lambda]$ which is an integer multiple of $\lambda$,
minimizing $|\{v\in V: \alpha \le x_v < \alpha+\lambda\}|$.  This can be found in $O(|\{v\in V: x_v \ge \frac{1}{2}-\gamma-\lambda\}|)\le O(\OPT)$ time.\footnote{
This assumes an appropriate encoding of the vector $(x_v)_{v\in V}$, for example, the
encoding from the proof of Lemma~\ref{lem:mwu}.
}

Partition $V$ into 3 subsets: $L=\{v\in V: x_v < \alpha\}$, $H=\{v\in V: x_v > 1-\alpha\}$, and $K=\{v\in V: \alpha \le x_v \le 1-\alpha\}$.
Note that $|K|\le \frac{1}{\alpha}\sum_{v\in V}x_v\le (2+O(\gamma))(1+\delta)\OPT = (2+O(\gamma))\OPT$.

Let $S_K$ be a $c$-approximate minimum vertex cover of $G[K]$. 
We claim that $S := S_K\cup H$ is a $c(1+O(\sqrt{\delta/\gamma}))$-approximation to the minimum vertex cover of $G$.

First, $S_K\cup H$ is a vertex cover of $G$, since vertices in $L$ can only be adjacent
to vertices in $H$.

Let $S^*$ be a minimum vertex cover of $G$, with $|S^*|=\OPT$.
Since $K\cap S^*$ is a vertex cover of $G[K]$, we have
$|S_K|\le c|K\cap S^*|$, and so $|S|\le c(|K\cap S^*| + |H|) = c(|S^*| + |H\setminus S^*| - |L\cap S^*| )$.

We will upper-bound $|H\setminus S^*|-|L\cap S^*|$.
To this end, let $L'=\{v\in V: \alpha \le x_v < \alpha+\lambda\}$, and define the following modified vector $(x_v')_{v\in V}$:
\[ x'_v = \left\{\begin{array}{ll}
              x_v-\lambda & \mbox{if $v\in H\setminus S^*$}\\
              x_v+\lambda & \mbox{if $v\in (L\cup L')\cap S^*$}\\
              x_v  & \mbox{otherwise.}
\end{array}\right.
\]
Note that $(x'_v)_{v\in V}$ is still a fractional vertex cover, since for each edge $uv\in E$ with $u\in H\setminus S^*$ and $v\in (L\cup L')\cap S^*$,
we have $x'_u+x'_v = (x_u-\lambda)+(x_v+\lambda)\ge 1$; on the other hand, for each edge $uv\in E$ with $u\in H\setminus S^*$ and
$v\not\in (L\cup L')\cap S^*$,
we have $v\in S^*$ (since $S^*$ is a vertex cover) and so $v\not\in L\cup L'$,
implying that $x'_u+x'_v \ge (1-\alpha-\lambda)+(\alpha+\lambda)\ge 1$.
Now, $\sum_{v\in V} x_v - \sum_{v\in V} x'_v = \lambda (|H\setminus S^*| - |L\cap S^*| - |L'\cap S^*|)$.
On the other hand, $\sum_{v\in V} x_v - \sum_{v\in V} x'_v\le (1-\frac{1}{1+O(\delta)})\sum_{v\in V}x_v
\le O(\delta)|S^*|$.  It follows that $|H\setminus S^*| - |L\cap S^*| \le O(\frac{\delta}{\lambda})|S^*| + |L'\cap S^*|$.

By our choice of $\alpha$, we have $|L'|\le O(\frac{1}{\gamma/\lambda}) \cdot |\{v\in V: x_v \ge \frac{1}{2}-\gamma-\lambda\}|
\le O(\frac{\lambda}{\gamma})|S^*|$.
We conclude that $|S|\le c(|S^*| + |H\setminus S^*| - |L\cap S^*| ) \le c(1+O(\frac\delta\lambda + \frac\lambda\gamma)) |S^*|$.  Choose $\lambda=\sqrt{\gamma\delta}$.
\end{proof}

\subsection{Dynamic Geometric Vertex Cover via Kernels}

We now use kernels to reduce the dynamic vertex cover for geometric intersection graphs to a special case of
static vertex cover where the number of objects is approximately at most $2\,\OPT$.
We use a simple, standard idea for dynamization: be lazy, and periodically recompute the solution only after every $\eps\,\OPT$ updates
(since the optimal size changes by at most 1 per update).
We observe that the data structure subproblem of dynamic min-weight intersecting pair, needed in Lemma~\ref{lem:mwu}, is reducible to 
dynamic intersection detection by known techniques.

\begin{lemma}\label{lem:Epp}
Let $\CC$ be a class of geometric objects, where there is a dynamic data structure $\DS_0$ for $n$ objects in $\CC$ that 
can detect whether there is an object 
intersecting a query object, and supports insertions and deletions of objects, with $O(\tau_0(n))$ query and update time.  

Then there is a dynamic data structure $\DS$ for $n$ weighted objects in $\CC$ that maintains
an intersecting pair of objects minimizing the sum of the weights, under insertions and deletions
of weighted objects, with $\OO(\tau_0(n))$ amortized time.
\end{lemma}
\begin{proof}
Eppstein~\cite{Eppstein95} gave a general technique to reduce the problem of dynamic closest pair to
the problem of dynamic nearest neighbor search, for arbitrary distance functions, while increasing
the time per operation by at most two logarithmic factors (with amortization).  Chan~\cite{Chan20} gave an alternative
method to achieve a similar result.  The  $\DS$ problem can be viewed as a dynamic closest pair problem, where the distance between
objects $u$ and $v$ is $w_u+w_v$ if they intersect, and $\infty$ otherwise.
Thus, our problem reduces to designing a data structure that can find a min-weight object intersecting
a query object, subject to insertions and deletions of objects (\emph{dynamic min-weight intersection searching}).

We can further reduce this to the $\DS_0$ problem (\emph{dynamic intersection detection}) by a standard \emph{multi-level} data structuring technique~\cite{AgarwalE99}, where
the primary data structure is a 1D search tree over the weights, 
and each node of the tree stores a secondary data structure for dynamic intersection detection.
Query and update time increase by one more logarithmic factor.
\end{proof}

\begin{theorem}\label{thm:main}
Let $c\ge 1$, $\eps>0$, and $0\le\delta<\gamma<\frac14$.
Let $\CC$ be a class of geometric objects with the following oracles:
\begin{enumerate}
\item[(i)] a dynamic data structure $\DS_0$ for $n$ objects in $\CC$ that
can detect whether there is an object 
intersecting a query object, and supports insertions and deletions of objects, with $O(\tau_0(n))$ query and update time;
\item[(ii)] a static algorithm $\AAA$ for computing a $c$-approximation
of the minimum vertex cover of the intersection graph of $n$ objects in $\CC$ in $T(n)$ time, under the promise that 
$n\le (2+O(\gamma))\,\OPT$, where $\OPT$ is the optimal vertex cover size.
\end{enumerate}

Then there is a dynamic data structure for $n$ objects in $\CC$ that maintains a $c(1+O(\sqrt{\delta/\gamma}+\eps))$-approximation
of the minimum vertex cover of the intersection graph, under insertions and deletions of objects,
in $\OO((1/(\delta^2\eps)) \tau_0(n) + (1/\eps)T(n)/n)$ amortized time, assuming that $T(n)/n$ is monotonically increasing and $T(2n)=O(T(n))$.
\end{theorem}
\begin{proof}
Assume that $b\le\OPT < 2b$ for a given parameter $b$.
Divide into phases with $\eps b$ updates each.
At the beginning of each phase:
\begin{enumerate}
\item Compute a $(1+\delta)$-approximation to the minimum fractional vertex cover
by Lemma~\ref{lem:mwu} in $\OO((1/\delta^2)b\cdot \tau_0(n))$ amortized time using the data structure $\DS$ from Lemma~\ref{lem:Epp}.
(A weight change can be done by a deletion and an insertion.  It is important to note that we don't rebuild $\DS$
at the beginning of each phase; we continue using the same structure $\DS$ in the next phase, after resetting
the modified weights back to~1.)
\item Generate a kernel $K$ with size at most $(2+O(\gamma))\,\OPT$ by Lemma~\ref{lem:kernel} in $O(b)$ time.
\item Compute a $c$-approximation to the minimum vertex cover of the intersection graph of $K$ by the static algorithm $\AAA$ in
$O(T((2+O(\gamma))\,\OPT))=O(T(b))$ time, from which we obtain a $c(1+O(\sqrt{\delta/\gamma}))$-approximation
of the minimum vertex cover of the entire intersection graph.
\end{enumerate}
Since the above is done only at the beginning of each phase, the amortized cost per update is
$\OO(\frac{(1/\delta^2)b\cdot \tau_0(n) + T(b)}{\eps b})$.

During a phase, we handle an object insertion simply by inserting it to the current cover, and we handle an object deletion
simply by removing  it from the current cover.  This incurs an additive error at most $O(\eps b)=O(\eps\,\OPT)$.
We also perform the insertion/deletion of the object in $\DS$, with initial weight 1, in $\OO(\tau_0(n))$ amortized time.

How do we obtain a correct guess $b$?  We build the above data structure
for each $b$ that is a power of 2, and run the algorithm simultaneously for each $b$ (with appropriate cap on the run time based on $b$).
\end{proof}

\subsection{Specific Results}\label{sec:specific}

We now apply our framework to solve the dynamic geometric vertex cover problem
for various specific families of geometric objects.  By Theorem~\ref{thm:main}, it suffices to provide
(i)~a dynamic data structure $\DS_0$ for intersection detection queries,
and (ii)~a static algorithm $\AAA$ for solving the special case of the vertex cover problem
when $n\le (2+O(\gamma))\OPT$.

\paragraph{Disks in $\R^2$.}
Intersection detection queries for disks in $\R^2$ reduce to additively weighted Euclidean nearest neighbor search,
where the weights (different from the weights from MWU) are the radii of the disks.
Kaplan et al.~\cite{KaplanKKKMRS22} adapted Chan's data structure~\cite{Chan20a, Chan01} for dynamic Euclidean nearest neighbor search in $\R^2$
and obtained a data structure for
additively weighted Euclidean nearest neighbor search  with polylogarithmic amortized update time and query time. Thus,
$\DS_0$ can be implemented in $\tau_0(n)=O(\log^{O(1)}n)$ amortized time for disks in~$\R^2$.

In Appendix~\ref{app:disks}, we give a static $(1+O(\eps))$-approximation algorithm $\AAA$ for MVC for disks, under the promise that $n=O(\OPT)$, with running time $T(n)=\OO(2^{O(1/\eps^2)}n)$.
The algorithm is obtained by modifying a known {\sf PTAS} for MIS for disks~\cite{Chan03}.

Applying Theorem~\ref{thm:main} with $c=1+O(\eps)$, $\gamma=\Theta(1)$, and $\delta=\eps^2$, we obtain our main result for disks:

\begin{corollary}\label{disk:MVC}
There is a dynamic data structure for $n$ disks in $\R^2$
that maintains a $(1+O(\eps))$-approximation of the minimum vertex cover
of the intersection graph,  under insertions and deletions,
in $O(2^{O(1/\eps^2)}\log^{O(1)}n)$ amortized time.
\end{corollary}

\paragraph{Rectangles in $\R^2$.}
For rectangles in $\R^2$, dynamic intersection detection requires $\tau_0(n)=O(\log^{O(1)}n)$ query and update time,
by standard orthogonal range searching techniques (range trees)~\cite{preparata2012computational,AgarwalE99}.

(Note: There is an alternative approach that bypasses Eppstein's technique and directly solves 
the $\DS$ data structure problem:  We dynamically maintain a \emph{biclique cover},
which can be done in polylogarithmic time for rectangles~\cite{Chan06}.  It is then easy to maintain the minimum weight
of each biclique, by maintaining the minimum weight of each of the two sides of the biclique with priority queues.)

In Appendix~\ref{app:rect}, we give 
a static $(\frac32 + O(\eps))$-approximation algorithm for MVC for rectangles, 
with running time $T(n)=\OO(2^{O(1/\eps^2)}n)$.  The algorithm is obtained by
modifying the method by Bar-Yehuda, Hermelin, and Rawitz~\cite{bar2011minimum}, and combining
with our efficient kernelization method.

Applying Theorem~\ref{thm:main} 
with $c=\frac32+O(\eps)$, $\gamma=\Theta(1)$, and $\delta=\eps^2$, 
we obtain our main result for rectangles:

\begin{corollary}\label{cor:rect}
There is a dynamic data structure for $n$ axis-aligned rectangles in $\R^2$
that maintains a $(\frac32+O(\eps))$-approximation of the minimum vertex cover
of the intersection graph, under insertions and deletions,
in $O(2^{O(1/\eps^2)}\log^{O(1)}n)$ amortized time.
\end{corollary}

\paragraph{Fat boxes (e.g., hypercubes) in $\R^d$.}

For the case of fat axis-aligned boxes (e.g., hypercubes) in a constant dimension~$d$,
dynamic intersection detection can again be solved with 
$\tau_0(n)=O(\log^{O(1)}n)$ amortized query and update time by orthogonal range searching~\cite{preparata2012computational,AgarwalE99}.
We can design the static algorithm $\AAA$ in exactly the same way as in the case of disks (since the method in Appendix~\ref{app:disks}
holds for fat objects in $\R^d$), achieving $T(n)=\OO(2^{O(1/\eps^d)}n)$.

\begin{corollary}\label{cor:fat:MVC}
There is a dynamic data structure for $n$ fat axis-aligned boxes in $\R^d$ for any constant $d$
that maintains a $(1+O(\eps))$-approximation of the minimum vertex cover
of the intersection graph, under insertions and deletions,
in $O(2^{O(1/\eps^d)}\log^{O(1)}n)$ amortized time.
\end{corollary}

We can also obtain results for balls or other types of fat objects in $\R^d$, but because intersection detection data structures
have higher complexity, the update time would be sublinear rather than polylogarithmic.

\paragraph{Bipartite disks in $\R^2$.}

For the case of a bipartite intersection graph between two sets of disks in $\R^2$,
we have $\tau_0(n)=O(\log^{O(1)}n)$ as already noted.
In the bipartite case, the static algorithm $\AAA$ is trivial: we just return the smaller of the two parts
in the bipartition, which yields a vertex cover of size at most $n/2\le (1+O(\eps))\,\OPT$ under the
promise that $n\le (2+O(\eps))\,\OPT$.

Applying Theorem~\ref{thm:main} with $c=1+O(\eps)$, $\gamma=\eps$, and $\delta=\eps^3$, we obtain:

\begin{corollary}\label{MVC:bipartite:disk}
There is a dynamic data structure for two sets of $O(n)$ disks in $\R^2$
that maintains a $(1+O(\eps))$-approximation of the minimum vertex cover
of the bipartite intersection graph, under insertions and deletions,
in $O((1/\eps^7)\log^{O(1)}n)$ amortized time.  
\end{corollary}

\paragraph{Bipartite boxes in $\R^d$.}

For the case of a bipartite intersection graph between two sets of (not necessarily fat) boxes in $\R^d$,
we have $\tau_0(n)=O(\log^{O(1)}n)$  by orthogonal range searching.
As already noted, in bipartite cases, the static algorithm $\AAA$ is trivial.

\begin{corollary}\label{MVC:bipartite:box}
There is a dynamic data structure for two sets of $O(n)$ axis-aligned boxes in $\R^d$ for any constant $d$
that maintains a $(1+O(\eps))$-approximation of the minimum vertex cover
of the bipartite intersection graph, under insertions and deletions,
in $O((1/\eps^7)\log^{O(1)}n)$ amortized time.  
\end{corollary}

We did not write out the number of logarithmic factors in our results, as we have not attempted to optimize them, but
it is upper-bounded by 3 plus the number of logarithmic factors 
in $\tau_0(n)$.

\section{Bipartite Maximum-Cardinality Matching}

The approach in the previous section finds a $(1+O(\eps))$-approximation of the MVC
of bipartite geometric intersection graphs, and so it allows us to approximate the size of the MCM in such bipartite graphs.  However, it doesn't compute a matching.  In this section,
we give a different approach to maintain a $(1+O(\eps))$-approximation of the MCM in bipartite intersection graphs.
The first thought that comes to mind is to compute a kernel, as we have done for MVC, but for MCM,
known approaches seem to yield only a kernel of $O(\OPT^2)$ size (e.g., see Gupta and Peng~\cite{GuptaP13}).  Instead, we will
bypass kernels and construct an approximate MCM directly.

\subsection{Approximate Bipartite MCM via Modified Hopcroft--Karp}

It is well known that Hopcroft and Karp's $O(m\sqrt{n})$-time algorithm for exact MCM in bipartite graphs~\cite{HopcroftK73}
can be modified to give a $(1+\eps)$-approximation algorithm that runs in near-linear time, simply
by terminating early after $O(1/\eps)$ iterations (for example, see the introduction in~\cite{DuanP10}).
We describe a way to reimplement the algorithm in sublinear time
when $\OPT$ is small by using appropriate data structures, which correspond to dynamic intersection detection in
the case of geometric intersection graphs.  
Note that earlier work by Efrat, Itai, and Katz~\cite{efrat2001geometry} has already combined Hopcroft and Karp's algorithm with geometric data structures to obtain static exact algorithm~\cite{efrat2001geometry} for
maximum matching in bipartite geometric intersection graphs (see also Har-Peled and Yang's paper~\cite{Har-PeledY22} on static approximation algorithms).
However, to achieve bounds sensitive to $\OPT$, our algorithm will work differently (in particular, it will be DFS-based instead of BFS-based).

\newcommand{\maximalaugpaths}{\textsc{maximal-aug-paths}}
\newcommand{\visit}{\textsc{visit}}
\newcommand{\extend}{\textsc{extend}}

\begin{lemma}\label{lem:mcm}
We are given an unweighted bipartite graph $G=(V,E)$.  Suppose there is a data structure $\DS_0$ for storing a subset $S\subseteq V$ of vertices, initially with $S=\emptyset$, that supports the following two operations in $\tau_0$ time: given a vertex $u\in V$, find a neighbor of $u$ that is in $S$ (if exists), and insert/delete a vertex to/from~$S$.

Given a data structure $\DS_0$ that currently has $S=V$, and given a maximal matching $M_0$,
we can compute a $(1+O(\eps))$-approximation to the maximum-cardinality matching in  $\OO((1/\eps^2)\,\OPT\cdot \tau_0)$ time.  Here, $\OPT$ denotes the maximum matching size.
\end{lemma}
\begin{proof}
The algorithm proceeds iteratively.  We maintain a matching $M$.
At the beginning of the $\ell$-th iteration, we know that the current matching $M$ does not
have augmenting paths of length $\le 2\ell-1$.
We find a maximal collection $\Gamma$ of vertex-disjoint augmenting paths
of length $2\ell+1$.  We then augment $M$ along the paths in $\Gamma$.
As shown by Hopcroft and Karp~\cite{HopcroftK73}, the new matching $M$ will then not have augmenting paths of length $\le 2\ell+1$.

As shown by Hopcroft and Karp~\cite{HopcroftK73}, there are at least $\OPT-|M|$ vertex-disjoint augmenting paths, and so
$|M|\ge (\ell+1)(\OPT-|M|)$, i.e., $\OPT\le (1+\frac{1}{\ell+1})|M|$.  Thus, once $\ell$ reaches $\Theta(1/\eps)$,
we may terminate the algorithm.  Initially, we can set $M=M_0$ before the first iteration.

It suffices to describe how to find a maximal collection $\Gamma$ of vertex-disjoint augmenting paths
of length $2\ell+1$ in the $\ell$-th iteration, under the assumption that there are no shorter augmenting paths.
Hopcroft and Karp originally proposed a BFS approach, starting at the ``exposed'' vertices not covered by the current matching $M$.  Unfortunately, this approach does not work in our setting:
because we want the running time to be near $\OPT$, the searches need to start at vertices of $M$.
We end up adopting a DFS approach, but the vertices of $M$ need to be duplicated $\ell$ times in $\ell$ ``layers'' (this increases the running time by a factor of $\ell$, but fortunately, $\ell$ is small in our setting). 

Let $V_M$ be the $2|M|$ vertices of the current matching $M$.  As is well known, $\OPT\le 2|M|$.
A walk $v_0u_1v_1\cdots u_\ell v_\ell u_{\ell+1}$ in $G$ is an \emph{augmenting walk} of length $2\ell+1$
if $v_0\not\in V_M$, $v_0u_1\not\in M$, $u_1v_1\in M$, \ldots, $v_\ell u_{\ell+1}\not\in M$, and $u_{\ell+1}\not\in V_M$.
In such an augmenting walk of length $\ell$, we automatically have $v_0\neq u_{\ell+1}$ (because $G$ is bipartite)
and the walk must automatically be a simple path (because otherwise we could short-cut and obtain an augmenting path of length
$\le 2\ell-1$).
In the procedure $\extend(v_0u_1v_1\cdots u_i,\ell)$ below, the input is a walk $v_0u_1v_1\cdots u_i$ 
with $v_0u_1\not\in M$, $u_1v_1\in M$, \ldots, $v_{i-1} u_i\not\in M$, and the output is true if it is possible
to extend it to an augmenting walk $v_0u_1v_1\cdots u_\ell v_\ell u_{\ell+1}$  of length $2\ell+1$ that
is vertex-disjoint from the augmenting walks generated so far.

\begin{quote}
\begin{tabbing}
9.\ \ \ \= for \= for \= for \=\kill
$\maximalaugpaths(\ell)$:\\[.25ex]
1.\> let $S=V\setminus V_M$ and $S_1=\cdots=S_\ell=V_M$\\
2.\> for each $u_1\in S_1$ do\\
3.\>\> let $v_0$ be a neighbor of $u_1$ with $v_0\in S$\\
4.\>\> if $v_0$ does not exist then delete $u_1$ from $S_1$\\
5.\>\> else $\extend(v_0u_1,\ell)$\\[2ex]
$\extend(v_0u_1v_1\cdots u_i,\ell)$:\\[.25ex]
1.\> let $v_i$ be the partner of $u_i$ in $M$\\
2.\> if $i=\ell$ then\\
3.\>\> let $u_{\ell+1}$ be a neighbor of $v_\ell$ with $u_{\ell+1}\in S$\\
4.\>\> if $u_{\ell+1}$ does not exist then delete $u_\ell$ from $S_\ell$ and return false\\
5.\>\> output the augmenting path $v_0u_1v_1\cdots u_\ell v_\ell u_{\ell+1}$\\
6.\>\> delete $v_0$ and $u_{\ell+1}$ from $S$, and $u_1,\ldots,u_\ell$ from all of $S_1,\ldots,S_\ell$, and return true\\
7.\> for each neighbor $u_{i+1}$ of $v_i$ with $u_{i+1}\in S_{i+1}\setminus\{u_i\}$ do\\
8.\>\> if $\extend(v_0u_1v_1\cdots u_{i+1},\ell)=$ true then return true\\
9.\> delete $u_i$ from $S_i$ and return false
\end{tabbing}
\end{quote}

Note that if it is not possible to extend the walk $v_0u_1v_1\cdots u_i$ to an augmenting walk of length $2\ell+1$,
then it is not possible to extend any other walk $v_0'u_1'v_1'\cdots u_i'$ of the same length with $u_i'=u_i$.
This justifies why we may delete $u_i$ from $S_i$ in line~9 of $\extend$ (and similarly why we may delete $u_\ell$ from $S_\ell$
in line~4 of $\extend$, and why we may delete $u_1$ from $S_1$ in line~4 of $\maximalaugpaths$).

For the running time analysis, note that each vertex may be a candidate for $u_i$ in only one call to $\extend$ per $i$,
because we delete $u_i$ from $S_i$, either in line~9 if false is returned, or in line~6  if true is returned.
Thus, the number of calls to $\extend$ is at most $O(\ell|M|)$.

We use the given data structure $\DS_0$ to maintain $S$.  This allows us to do line~3 of $\maximalaugpaths$.
Line~1 of $\maximalaugpaths$ requires $O(|M|)$ initial deletions from $S$.  At the end, we 
reset $S$ to $V$ by performing $O(|M|)$ insertions of the deleted elements.

We also maintain $S_1,\ldots,S_\ell$ in $\ell$ new instances of the data structure $\DS_0$.
This allows us to do lines 3 and 7 of $\extend$.  Line~1 of $\maximalaugpaths$ requires $O(\ell|M|)$ initial
insertions to $S_1,\ldots,S_\ell$.
We conclude that $\maximalaugpaths$ takes $O(\ell |M|\cdot\tau_0)$ time.
Hence, the overall running time of all $\ell=\Theta(1/\eps)$ iterations is $O((1/\eps^2)|M|\cdot\tau_0)$.
\IGNORE{
\begin{quote}
\begin{tabbing}
9.\ \ \= for \= for \= for \=\kill
$\maximalaugpaths(\ell)$:\\[.25ex]
1.\> let $S=V\setminus V_M$ and $S_1=\cdots=S_\ell=V_M$\\
2.\> for each $u_1\in S_1$ do\\
3.\>\> let $v_0$ be a neighbor of $u_1$ with $v_0\in S$\\
4.\>\> if $v_0$ does not exist then delete $u_1$ from $S_1$\\
5.\>\> else if $\visit(u_1,1)=$ true then output $v_0u_1$, delete $v_0$ from $S$, and return true\\[2ex]
$\visit(u_i,i)$:\\[.25ex]
1.\> let $v_i$ be the partner of $u_i$ in $M$\\
2.\> if $i=\ell$ then\\
3.\>\> let $u_{\ell+1}$ be a neighbor of $v_\ell$ with $u_{\ell+1}\in S$\\
4.\>\> if $u_{\ell+1}$ does not exist then delete $u_\ell$ from $S_\ell$ and return false\\
5.\>\> output $u_\ell v_\ell$ and $v_\ell u_{\ell+1}$, delete $u_{\ell+1}$ from $S$ and $u_\ell$ from $S_1,\ldots,S_\ell$, and return true\\
6.\> for each neighbor $u_{i+1}$ of $v_i$ with $u_{i+1}\in S_{i+1}\setminus\{u_i\}$ do\\
7.\>\> if $\visit(u_{i+1},i+1)=$ true then output $u_i v_i$ and $v_i u_{i+1}$, delete $u_i$ from $S_1,\ldots,S_\ell$, and return true\\
8.\> delete $u_i$ from $S_i$ and return false
\end{tabbing}
\end{quote}
}
\end{proof}

\subsection{Dynamic Geometric Bipartite MCM} 

To solve dynamic bipartite MCM for geometric intersection graphs, we again use a standard idea for dynamization: be lazy, and periodically recompute the solution only after every $\eps\,\OPT$ updates
(since the optimal size changes by at most 1 per update).

\begin{theorem}\label{thm:main:mcm}
Let $\CC$ be a class of geometric objects, where there is
a dynamic data structure $\DS_0$ for $n$ objects in $\CC$ that
can find an object 
intersecting a query object (if exists), and supports insertions and deletions of objects, with $O(\tau_0(n))$ query and update time.

Then there is a dynamic data structure for two sets of $O(n)$ objects in $\CC$ that maintains a $(1+O(\eps))$-approximation
of the maximum-cardinality matching in the bipartite intersection graph, under insertions and deletions of objects,
in $\OO((1/\eps^3) \tau_0(n))$ amortized time.
\end{theorem}
\begin{proof}
First, observe that we can maintain a maximal matching $M_0$ with $O(\tau_0(n))$ update time:
We maintain the subset $S$ of all vertices not in $M_0$ in a data structure $\DS_0$.
When we insert a new object $u$, we match it with an object in $S$ intersecting $u$ (if exists) by querying $\DS_0$.
When we delete an object $u$, we delete $u$ and its partner $v$ in $M_0$, and reinsert $v$.

Assume that $b\le\OPT < 2b$ for a given parameter $b$.
Divide into phases with $\eps b$ updates each.
At the beginning of each phase, compute a $(1+O(\eps))$-approximation of the maximum-cardinality matching
by Lemma~\ref{lem:mcm} in $\OO((1/\eps^2)b\cdot \tau_0(n))$ time.
(It is important to note that we don't rebuild the data structure $\DS$ for $S=V$
at the beginning of each phase; we continue using the same structure $\DS$ in the next phase, after resetting
the modified $S$ back to $V$.)

Since the above is done only at the beginning of each phase, the amortized cost per update is
$\OO(\frac{(1/\eps^2)b\cdot \tau_0(n)}{\eps b})$.

During a phase, we handle an object insertion simply by doing nothing, and we handle an object deletion
simply by removing its incident edge (if exists) from the current matching.  This incurs an additive error at most $O(\eps b)=O(\eps\,\OPT)$.
We also perform the insertion/deletion of the object in $\DS_0$ for $S=V$, in $\OO(\tau_0(n))$ time.

How do we obtain a correct guess $b$?  We build the above data structure
for each $b$ that is a power of 2, and run the algorithm simultaneously for each $b$.
\end{proof}

\subsection{Specific Results}

Recall that for disks in $\R^2$ as well as boxes in $\R^d$,
we have $\tau_0(n)=O(\log^{O(1)}n)$.
(Note that most data structures for intersection detection can be modified to report a witness
object intersecting the query object if the answer is true.)
Thus, we immediately obtain:

\begin{corollary}\label{cor:mcm:disk}
There is a dynamic data structure for two sets of $O(n)$ disks in $\R^2$
that maintains a $(1+O(\eps))$-approximation of the maximum-cardinality matching
in the bipartite intersection graph, under insertions and deletions,
in $O((1/\eps^3)\log^{O(1)}n)$ amortized time.  
\end{corollary}

\begin{corollary}\label{cor:mcm:box}
There is a dynamic data structure for two sets of $O(n)$ axis-aligned boxes in $\R^d$ for any constant $d$
that maintains a $(1+O(\eps))$-approximation of the maximum-cardinality matching
in the bipartite intersection graph, under insertions and deletions,
in $O((1/\eps^3)\log^{O(1)}n)$ amortized time.  
\end{corollary}

\newcommand{\ZZ}{\mathcal{Z}}
\newcommand{\HH}{\mathcal{H}}

\newcommand{\GENERALMATCHINGSECTION}{

\section{General Maximum-Cardinality Matching}\label{sec:gen:match}

In this section, we adapt our results for bipartite MCM to
general MCM\@.  We use a simple idea
by Lotker, Patt-Shamir, and Pettie~\cite{LotkerPP15} to reduce the problem for general graphs (in the
approximate setting) to finding maximal collections of short augmenting paths in \emph{bipartite} graphs, which we already know how to solve.
(See also Har-Peled and Yang's paper~\cite{Har-PeledY22}, which applied Lotker et al.'s idea to obtain static approximation algorithms for  geometric intersection graphs.)  
The reduction has exponential dependence in the path length $\ell$ (which is fine since $\ell$ is small), and is originally randomized.
We reinterpret their idea in terms of \emph{color-coding}~\cite{AlonYZ95},
which allows for efficient derandomization, and also simplifies the analysis (bypassing Chernoff-bound calculations).  With this reinterpretation, it is easy to show that the idea carries over to the
dynamic setting.

We begin with a lemma, which is a consequence of the standard color-coding technique:

\begin{lemma}\label{lem:color:cod}
For any $n$ and $\ell$, there exists a collection $\ZZ^{(n,\ell)}$ of $O(2^{O(\ell)}\log n)$ subsets
of $[n] := \{1,\ldots,n\}$ such that
for any two disjoint sets $A,B\subseteq [n]$ of total size at most $\ell$, we have 
$A\subseteq Z$ and $B\subseteq [n]\setminus Z$ for some $Z\in\ZZ$.
Furthermore, $\ZZ^{(n,\ell)}$ can be constructed in $O(2^{O(\ell)}n\log n)$ time.
\end{lemma}
\begin{proof}
As shown by Alon, Yuster, and Zwick~\cite{AlonYZ95}, there exists a collection $\HH^{(n,\ell)}$
of $O(2^{O(\ell)}\log n)$ mappings $h:[n]\rightarrow [\ell]$,
such that for any set $S\subseteq [n]$ of size at most $\ell$,
the elements in $\{h(v): v\in S\}$ are all distinct.  Furthermore, $\HH^{(n,\ell)}$ can be constructed
in $O(2^{O(\ell)}n\log n)$ time.  (This is related to the notion of ``$\ell$-perfect hash family''.)

For each $h\in\HH^{(n,\ell)}$ and each subset $I\subseteq [\ell]$,
add the subset $Z_{h,I}=\{v\in[n]: h(v)\in I\}$ to $\ZZ^{(n,\ell)}$.
The number of subsets is $|\HH^{(n,\ell)}|\cdot 2^\ell \le 2^{O(\ell)}\log n$.
For any two disjoint sets $A,B\subseteq[n]$ of total size at most~$\ell$,
let $h\in\HH^{(n,\ell)}$ be such that the elements in $\{h(a): a\in A\}$ and $\{h(b): b\in B\}$
are all distinct,
and let $I=\{h(a): a\in A\}$; then $A\subseteq Z_{h,I}$
and $B\subseteq [n]\setminus Z_{h,I}$.
\end{proof}

We now present a non-bipartite analog of Lemma~\ref{lem:mcm}:

\begin{lemma}\label{lem:mcm:gen}
We are given an unweighted graph $G=(V,E)$, with $V=[n]$.  Let $\ZZ^{(n,1/\eps)}$ be as in Lemma~\ref{lem:color:cod}.
 Suppose there is a data structure $\DS_0^*$ for storing a subset $S\subseteq V$ of vertices, initially with $S=\emptyset$, that supports the following two operations in $\tau_0$ time: given a vertex $u\in V$ and $Z\in\ZZ^{(n,1/\eps)}$, find a neighbor of $u$ that is in $S\cap Z$
(if exists) and a neighbor of $u$ that is in $S\setminus Z$ (if exists); and insert/delete a vertex to/from~$S$.

Given a data structure $\DS_0$ that currently has $S=V$, and given a maximal matching $M_0$,
we can compute a $(1+O(\eps))$-approximation to the maximum-cardinality matching in  $\OO(2^{O(1/\eps)}\,\OPT\cdot \tau_0)$ time.  Here, $\OPT$ denotes the maximum matching size.
\end{lemma}
\begin{proof}
As in the proof of Lemma~\ref{lem:mcm}, we iteratively maintain a current matching $M$,
and it suffices to describe how to find a maximal collection $\Gamma$ of vertex-disjoint augmenting paths
of length $2\ell+1$ in the $\ell$-th iteration, under the assumption that there are no augmenting paths of length
$\le 2\ell-1$.
However, the presence of odd-length cycles complicates the computation of $\Gamma$.

Initialize $\Gamma=\emptyset$.  We loop through each $Z\in \ZZ^{(n,1/\eps)}$ one by one and do the following.  
Let $G_Z$ be the subgraph of $G$ with edges $\{uv\in E: u\in Z,\, v\not\in Z\}$.
We find a maximal collection of vertex-disjoint augmenting paths of length $2\ell+1$ in $G_Z$ that
are vertex-disjoint from paths already selected to be in $\Gamma$; we then add this new collection to $\Gamma$.
Since $G_Z$ is bipartite, this step can be done using the $\maximalaugpaths$ procedure from the
proof of Lemma~\ref{lem:mcm}.  Since we are working with $G_Z$ instead of $G$, when we find neighbors of a given vertex,
they are now restricted to be in $Z$ if the given vertex is in $[n]\setminus Z$, or vice versa; the data structure
$\DS_0^*$ allows for such queries.
The only other change is that when we initialize $S,S_1,\ldots,S_\ell$, we should remove vertices that have 
appeared in paths already selected to be in~$\Gamma$.  

Assume $2\ell+2 \le 1/\eps$.
We claim that after looping through all $Z\in \ZZ^{(n,1/\eps)}$, the resulting collection $\Gamma$ of vertex-disjoint
length-$(2\ell+1)$ augmenting paths is maximal in $G$.  To see this, consider any length-$(2\ell+1)$ augmenting path
$v_0u_1v_1\cdots u_\ell v_\ell u_{\ell+1}$  in $G$.  There exists $Z\in \ZZ^{(n,1/\eps)}$ such that
$u_1,\ldots,u_{\ell+1}\in Z$ and $v_0,\ldots,v_\ell\not\in Z$.
Thus, the path must intersect some path in $\Gamma$ during the iteration when we consider $Z$.
\end{proof}

We can now obtain a non-bipartite analog of Theorem~\ref{thm:main:mcm}:

\begin{theorem}\label{thm:main:mcm:gen}
Let $\CC$ be a class of geometric objects, where there is
a dynamic data structure $\DS_0$ for $n$ objects in $\CC$ that
can find an object 
intersecting a query object (if exists), and supports insertions and deletions of objects, with $O(\tau_0(n))$ query and update time.

Then there is a dynamic data structure for $O(n)$ objects in $\CC$ that maintains a $(1+O(\eps))$-approximation
of the maximum-cardinality matching in the intersection graph, under insertions and deletions of objects,
in $\OO(2^{O(1/\eps)} \tau_0(n))$ amortized time.
\end{theorem}
\begin{proof}
This is similar to the proof of Theorem~\ref{thm:main:mcm}, with Lemma~\ref{lem:mcm:gen} replacing
Lemma~\ref{lem:mcm}.  The only difference is that to support the data structure $\DS_0^*$, we maintain
$O(2^{O(1/\eps)}\log n)$ parallel instances of the data structure $\DS_0$ for $S\cap Z$ and
$S\setminus Z$, for every $Z\in \ZZ^{(n,1/\eps)}$.
This increases the update time by a factor of $O(2^{O(1/\eps)}\log n)$.

We have assumed that the input objects are labeled by integers in $[n]$.  When a new object is inserted, we can just
assign it the next available label in $[n]$.  When the number of objects exceeds $n$, we double $n$ and rebuild
the entire data structure from scratch.  Similarly, when the number of objects is below $n/4$, we halve $n$ and rebuild.
\end{proof}

\begin{corollary}\label{cor:mcm:disk2}
There is a dynamic data structure for $n$ disks in $\R^2$
that maintains a $(1+O(\eps))$-approximation of the maximum-cardinality matching
in the intersection graph, under insertions and deletions,
in $O(2^{O(1/\eps)}\log^{O(1)}n)$ amortized time.  
\end{corollary}

\begin{corollary}\label{cor:mcm:box2}
There is a dynamic data structure for $n$ axis-aligned boxes in $\R^d$ for any constant $d$
that maintains a $(1+O(\eps))$-approximation of the maximum-cardinality matching
in the intersection graph, under insertions and deletions,
in $O(2^{O(1/\eps)}\log^{O(1)}n)$ amortized time.  
\end{corollary}

}

\ARXIV{\GENERALMATCHINGSECTION}
\ICALP{
In Appendix~\ref{sec:gen:match}, we adapt our results for bipartite MCM to general MCM.
}

\section{Conclusion}

In this paper, we studied two fundamental graph optimization problems, MVC and MCM, in fully dynamic geometric settings. For both of the problems, we developed general frameworks based on the dynamic intersection detection data structures and static approximation algorithms, which are necessary ingredients for designing dynamic algorithms for geometric intersection graphs. These frameworks allowed us to obtain dynamic algorithms for a wide range of geometric intersection graphs, e.g., disks/rectangles in $\mathbb{R}^2$, hypercubes in $\mathbb{R}^d$, etc. Moreover, our results extend to the bipartite versions of the problems as well. 
In this work, we primarily focused on the unweighted case. Our approach does not readily extend to the weighted setting. Obtaining efficient dynamic algorithms for both MVC and MCM in the weighted setting are interesting open problems. 






\old{
}


\bibliographystyle{alphaurl}
\bibliography{main.bib}

\newpage

\appendix

\section{Minimum Vertex Cover: Fast Static Algorithms}

\subsection{Disks in $\R^2$}\label{app:disks}

In this subsection,
we give a static, near-linear-time, $(1+O(\eps))$-approximation algorithm $\AAA$ for MVC for disks under the promise that $n=O(\OPT)$.
As we are aiming for a $(1+O(\eps))$-factor approximation, 
we may tolerate an additive error of $O(\eps n)$
because of the promise.  MVC with $O(\eps n)$ additive error reduces
to MIS with $O(\eps n)$ additive error, by complementing the solution.

{\sf PTAS}s are known for MIS for disks, but allowing $O(\eps n)$ additive error, there are
actually {\sf EPTAS}s that run in near linear time.  For example, we can adapt an approach by Chan~\cite{Chan03}
based on divide-and-conquer via separators.
Specifically, we can use the following variant of Smith and Wormald's geometric separator theorem~\cite{SmithW98}:

\begin{lemma}[Smith and Wormald's separator]\label{lem:sep}
Given $n$ fat objects in a constant dimension $d$, there is an axis-aligned hypercube $B$, such that the
number of objects inside $B$ and the number of objects outside $B$ are both at most $(1-\beta)n$ for some constant $\beta>0$
(dependent on $d$),
and the objects intersecting $\partial B$ can be stabbed by $O(n^{1-1/d})$ points.  Furthermore, $B$ can be constructed in $O(n)$ time.
\end{lemma}
\begin{proof}
(The following is a modification of a proof in~\cite{ChanH24}, which is based on Smith and Wormald's original proof~\cite{SmithW98}.)

For each object, first pick an arbitrary ``reference'' point inside the object.
Let $B_0$ be the smallest hypercube that contains at least $\frac{n}{2^d+1}$ 
reference points.  Let $r$ be the side length of $B_0$.
Let $h$ be a parameter to be set later.
For each $t\in\{\frac 1h, \frac2h,\ldots, \frac{h-1}{h}\}$,
let $B_t$ be the hypercube with a scaled copy of $B_0$
with the same center and side length $(1+t)r$.
Since $B_t$ can be covered by $2^d$ quadtree boxes of side length $<r$,
the number of reference points inside $B_t$ is at most $\frac{2^d n}{2^d+1}$.

An object of diameter $\le r/h$ intersects $\partial B_t$ for at most $O(1)$ choices of $t$.
Thus, we can find a value $t$ in $O(n)$ time such that $\partial B_t$ intersects $O(n/h)$ objects
of diameter $\le r/h$.  On the other hand, the objects of diameter $> r/h$ intersecting $\partial B_t$
can all be stabbed by $O(h^{d-1})$ points, because of fatness.  Set $h=n^{1/d}$, and $B_t$ satisfies the desired properties
with $\beta=\frac{2^d}{2^d+1}$.

One remaining issue is how to find $B_0$ quickly. 
Let $C$ be a sufficiently large constant.  Form a $C\times\cdots\times C$ grid, so that the number of reference points between any two consecutive parallel grid hyperplanes is $n/C$.  This takes $O(C n)$ time by 
a linear-time selection algorithm.  Round each reference point to its nearest grid point.  The rounded
point set is a multiset with $O(C^d)$ distinct elements.
Redefine $B_0$ as the smallest hypercube that contains at most $\frac{n}{2^d+1}$ rounded reference points (multiplicity included),
which can now be computed in $C^{O(1)}$ time.  For any box, the number of reference points inside the box
changes by at most $O(n/C)$ after rounding.  So, we just need to increase $\beta$ by $O(1/C)$.
\IGNORE{
For each object, first pick a ``reference'' point inside the object.
Compute the smallest quadtree box $B^*$ (a hypercube) that contains at least $2^d\beta n$ 
reference points; this can be found in $\OO(n)$ time by constructing and traversing the
(compressed) quadtree for the $n$ reference points.
Now, $B^*$ has $2^d$ child boxes (of half the diameter), and one of the child boxes, denoted by
$B_0$, contains at least $\beta n$ reference points.  Let $r$ be the diameter of $B_0$.
Let $b$ be a parameter to be set later.
For $t\in\{\frac 1b, \frac2b,\ldots, \frac{b-1}{b}$,
let $B_t$ be the hypercube with a scaled copy of $B_0$
with the same center and diameter $(1+t)r$.
Since $B_t$ can be covered by $3^d$ quadtree boxes congruent to $B_0$,
we see that the number of reference points inside $B_t$ is $6^d\beta n=(1-\beta)n$,
by choosing $\beta_d=1/(6^d+1)$.

An object of diameter $\le r/b$ intersects $\partial B_t$ for at most $O(1)$ choices of $t$.
Thus, we can find some $t$ in $O(n)$ time such that $\partial B_t$ intersects $O(n/b)$ objects
of diameter $\le r/b$.  On the other hand, the objects of diameter $> r/b$ intersecting $\partial B_t$
can be stabbed by $O(b^{d-1})$ points because of fatness.  Set $b=n^{1/d}$, and $B_t$ satisfies the desired properties.
}    
\end{proof}

The MIS algorithm is simple: we just recursively compute an independent set for the disks inside $B$
and an independent set for the disks outside $B$, and take the union of the two sets.
If $n$ is below a constant $b$, we solve the problem exactly by brute force in $2^{O(C)}$ time.
(This is analogous to Lipton and Tarjan's original {\sf EPTAS} for independent set for planar graphs~\cite{LiptonT77}.)

Since the optimal independent set can have at most $O(\sqrt{n})$
disks intersecting $\partial B$, the total additive error satisfies a recurrence of the form
\[ E(n) \le  \left\{\begin{array}{ll}
 \max_{n_1,n_2\le (1-\beta)n:\ n_1+n_2\le n}(E(n_1)+E(n_2)+O(\sqrt{n})) & \mbox{if $n\ge b$}\\
 0 & \mbox{if $n<b$.}
 \end{array}\right.\]
This yields $E(n)=O(n/\sqrt{b})$.  We set $b=1/\eps^2$.
The running time of the resulting static algorithm $\AAA$ is $T(n)=\OO(2^{O(1/\eps^2)}n)$.

\subsection{Rectangles in $\R^2$}\label{app:rect}

In this subsection, we give a static, near-linear-time, $(\frac32+O(\eps))$-approximation algorithm for MVC for rectangles, by adapting Bar-Yehuda, Hermelin, and Rawitz's previous, polynomial-time, $(\frac32+O(\eps))$-approximation algorithm~\cite{bar2011minimum}.

\paragraph{Triangle-free case.}
We first start with the case when the intersection graph is triangle-free, or equivalently, when
the maximum depth of the rectangles is at most~2.  (The depth of a point $q$ among a set of rectangles refers to 
the number of rectangles containing $q$; the maximum depth refers to the maximum over all $q\in\R^2$.)

We will design a static algorithm $\AAA$ for MVC for rectangles, under 
the promise that $n\le (2+O(\eps))\,\OPT$.  


We say that a rectangle $r$ \emph{dominates} another rectangle $r'$ if
$\partial r$ intersects $\partial r'$ four times, with $r$ having larger height than $r'$.
Bar-Yehuda, Hermelin, and Rawitz's algorithm proceeds as follows:

\begin{enumerate}
\item Decompose $R$ into 2 subsets $R_1$ and $R_2$, such that there are no dominating pairs within each subset~$R_i$.

The existence of such a decomposition follows easily from Dilworth's theorem, but for an explicit construction,
 we can just define $R_1$ to contain all rectangles in $R$ that are not dominated
by any other rectangle in $R$, and define $R_2$ to be $R\setminus R_1$.
Correctness is easy to see (since the maximum depth is~2).

We can compute $R_1$ (and thus $R_2$) by performing $O(n)$ orthogonal range queries, after lifting each rectangle to 
a point in $\R^4$.  In fact, computing $R_1$ corresponds to the \emph{maxima} problem in $\R^4$, for which
$O(n\log n)$-time algorithms are known~\cite{ChanLP11}.

\item For each $i\in\{1,2\}$, compute a vertex cover $S_i$ of $R_i$ which approximates the minimum with additive
error at most $\eps n$.

Since we tolerate $\eps n$ additive error, it suffices to approximate the MIS of $R_i$
with $\eps n$ additive error.

Because $R_i$ has no dominating pairs, it forms a \emph{pseudo-disk} arrangement
(each pair intersects at most twice).
It is known~\cite{bar2011minimum, KedemLPS86} that the intersection graph of any set of pseudo-disks
that have maximum depth 2 and have no containment pair is in fact planar.
(We can easily eliminate containment pair, by removing rectangles that contain another rectangle;
and we can detect such rectangles again by orthogonal range queries.)
So, we can use known near-linear-time {\sf EPTAS} for MIS for planar graphs; for example, Lipton and Tarjan's
divide-and-conquer algorithm via planar-graph separators runs in $\OO(2^{O(1/\eps^2)}n)$ time~\cite{LiptonT77}.

\item Return $S$, the smaller of the two sets: $S_1\cup R_2$ and $S_2\cup R_1$.
\end{enumerate}

Let $S^*$ be the minimum vertex cover.
Then 
\begin{eqnarray*}
 |S| &\le& \tfrac{1}{2}(|S_1|+|R_2| + |S_2|+|R_1|)\\
 &\le& \tfrac{1}{2} (|S^*\cap R_1| + |S^*\cap R_2| + |R_1|+|R_2|) + \eps n
 \ =\ \tfrac{1}{2} (|S^*| + n) + \eps n\ \le\ (\tfrac32+O(\eps))|S^*|
\end{eqnarray*}
under the assumption that $n\le (2+O(\eps))|S^*|$.  The running time is
$T(n)=\OO(2^{O(1/\eps^2)}n)$.

At this point, if we apply Theorem~\ref{thm:main} with $c=\frac32+O(\eps)$, $\gamma=\eps$, and $\delta=\eps^3$ (and $\tau_0(n)=O(\log^{O(1)}n)$ as noted in Section~\ref{sec:specific}), we obtain:

\begin{corollary}\label{cor:rect:trifree}
There is a dynamic data structure for $n$ rectangles in $\R^2$
that maintains a $(\frac32+O(\eps))$-approximation of the minimum vertex cover
of the intersection graph, under insertions and deletions,
in $O(2^{O(1/\eps^2)}\log^{O(1)}n)$ amortized time, under the assumption that
the intersection graph is triangle-free.
\end{corollary}

\paragraph{General case.}
We now design our final static algorithm $\AAA$ for MVC for rectangles that avoids the triangle-free assumption, again by building on
Bar-Yehuda,  Hermelin, and Rawitz's approach:

\begin{enumerate}
\item Remove vertex-disjoint triangles $T_1,\ldots,T_\ell$ so that
the remaining set $R' := R\setminus (T_1\cup\cdots \cup T_\ell)$ is
triangle-free.

We can implement this step greedily in polynomial time, but a faster approach is via a plane sweep.
Namely, we modify the standard sweep-line algorithm to computing the maximum depth of $n$ rectangles in $\R^2$ (similar to Klee's measure problem)~\cite{preparata2012computational}.
The algorithm uses a data structure for a 1D problem: maintaining the maximum depth of intervals in $\R^1$,
subject to insertions and deletions of intervals.  Simple modification of standard search trees achieve
$O(\log n)$ time per insertion and deletion.  As we sweep a vertical line $\ell$ from left to right, we maintain
the maximum depth of the $y$-intervals of the rectangles intersected by $\ell$.  When $\ell$ hits
the left side of a rectangle, we insert an interval.  When $\ell$ hits the right side of a rectangle,
we delete an interval.  When the maximum depth becomes 3, we remove the 3 intervals containing the 
maximum-depth point, which correspond to a triangle in the intersection graph, and continue the sweep.
The total time is $O(n\log n)$.

\item Compute a vertex cover $S'$ of $R'$ which is a $(\frac32+O(\eps))$-approximation to the minimum.
This can be done by Corollary~\ref{cor:rect:trifree} in $\OO(2^{O(1/\eps^2)}n)$ time (we actually don't need
the full power of a dynamic data structure, since we just want a static algorithm for this step).

\item Return $S = S'\cup T$ with $T := T_1\cup\cdots\cup T_\ell$.
\end{enumerate}

Let $S^*$ be the minimum vertex cover.  
The key observation is that $S^*$ must contain at least 2 of the 3 vertices in each triangle $T_i$,
and so $|S^*\cap T|\ge \frac23 |T|$.
Thus,
\[ |S|\ =\ |S'|+|T|\ \le\ (\tfrac32+O(\eps))|S^*\cap R'| + \tfrac32|S^*\cap T|\ \le\ (\tfrac32+O(\eps))|S^*|.
\]
The running time of our final static algorithm $\AAA$ is $T(n)=\OO(2^{O(1/\eps^2)}n)$.


\ICALP{\GENERALMATCHINGSECTION}

\end{document}